\documentclass[conference,letter,10pt]{IEEEtran}

\usepackage{amsmath}
\usepackage{amsfonts}
\usepackage{amssymb}

\usepackage{theorem}
\usepackage{graphicx}
\usepackage{psfrag}
\usepackage{enumerate}

\usepackage{graphicx}
\usepackage{caption}
\usepackage{subcaption}
\usepackage{color}

\newtheorem{theorem}{Theorem}
\newtheorem{lemma}{Lemma}

\newtheorem{definition}{Definition}

\newtheorem{example}{Example}

\newcommand{\comment}[1]{}

\def\cX{\mbox{$\cal{X}$}}
\def\cY{\mbox{$\cal{Y}$}}
\def\cZ{\mbox{$\cal{Z}$}}

\def\cR{\mbox{$\cal{R}$}}

\title{Zero-Error Function Computation through a Bidirectional Relay}
\author{
\IEEEauthorblockN{Jithin~Ravi and Bikash~Kumar~Dey}
\IEEEauthorblockA{Department of Electrical Engineering \\Indian Institute of Technology Bombay \\
        {\tt \{rjithin,bikash\}@ee.iitb.ac.in}
}
}

\begin{document}
\maketitle
\begin{abstract}
We consider zero error function computation in a three node wireless network. Nodes A and B observe
$X$ and $Y$ respectively, and want to compute a function $f(X,Y)$ with zero error. To achieve this, nodes A and B
send messages to a relay node C at rates $R_A$ and $R_B$ respectively. 
The relay C then broadcasts a message to A and B at rate $R_C$ to help them 
compute $f(X,Y)$ with zero error. We allow block coding, and study the 
region of rate-triples $(R_A,R_B,R_C)$ that are feasible.
The rate region is characterized in terms of
graph coloring of some suitably defined probabilistic graphs. We
give single letter inner and outer bounds which meet for some
simple examples.
We provide a sufficient condition on the joint distribution $p_{XY}$ 
under which the relay can also compute
$f(X,Y)$ if A and B can compute it with zero error.
\end{abstract}

\section{Introduction}
Distributed computation of distributed data is a common problem
in a network. Such problems in various flavours have attracted strong 
research interest in the last decade. Gathering all the data at the
nodes where a function needs to be computed is wasteful in most situations.
So intermediate nodes also help by doing some processing of the data
to reduce the communication load on the links. Such computation frameworks
are known as distributed function computation or in-network function
computation \cite{Rai_2012,Shah_2013,MishraDPDisit14,Kowshik_2012}. 

We consider the problem of function computation in a wireless network with 
three nodes as shown in Fig.~\ref{Relay_Network}.
Nodes A and B have two correlated random variables $X$ and $Y$ respectively.
They have infinite i.i.d. realizations of these random variables.
They can communicate directly to a relay node C over
orthogonal error-free links. The relay node C can broadcast a message to
both A and B. A and B receive such broadcasted message
without error. 
Nodes A and B want to compute a function $f(X,Y)$ with zero error
for all realizations of $(X,Y)$ with nonzero probability.  We allow block coding of arbitrarily
large block length $n$. For each block of $n$ data symbols, we allow
two phases of communication. In the first phase, both A and B
send individual messages to C over the respective orthogonal links.
In the second phase, the relay broadcasts a message to A and B.
We study the expected number of bits that need to be sent per computation over
the individual links (A,C), (B,C) and over the broadcast link from
C to A and B.
Since all the nodes transmit once in our protocol, we call this a one-round 
protocol. 

\begin{figure}[htbp]
\centering
\includegraphics[width=0.7\columnwidth]{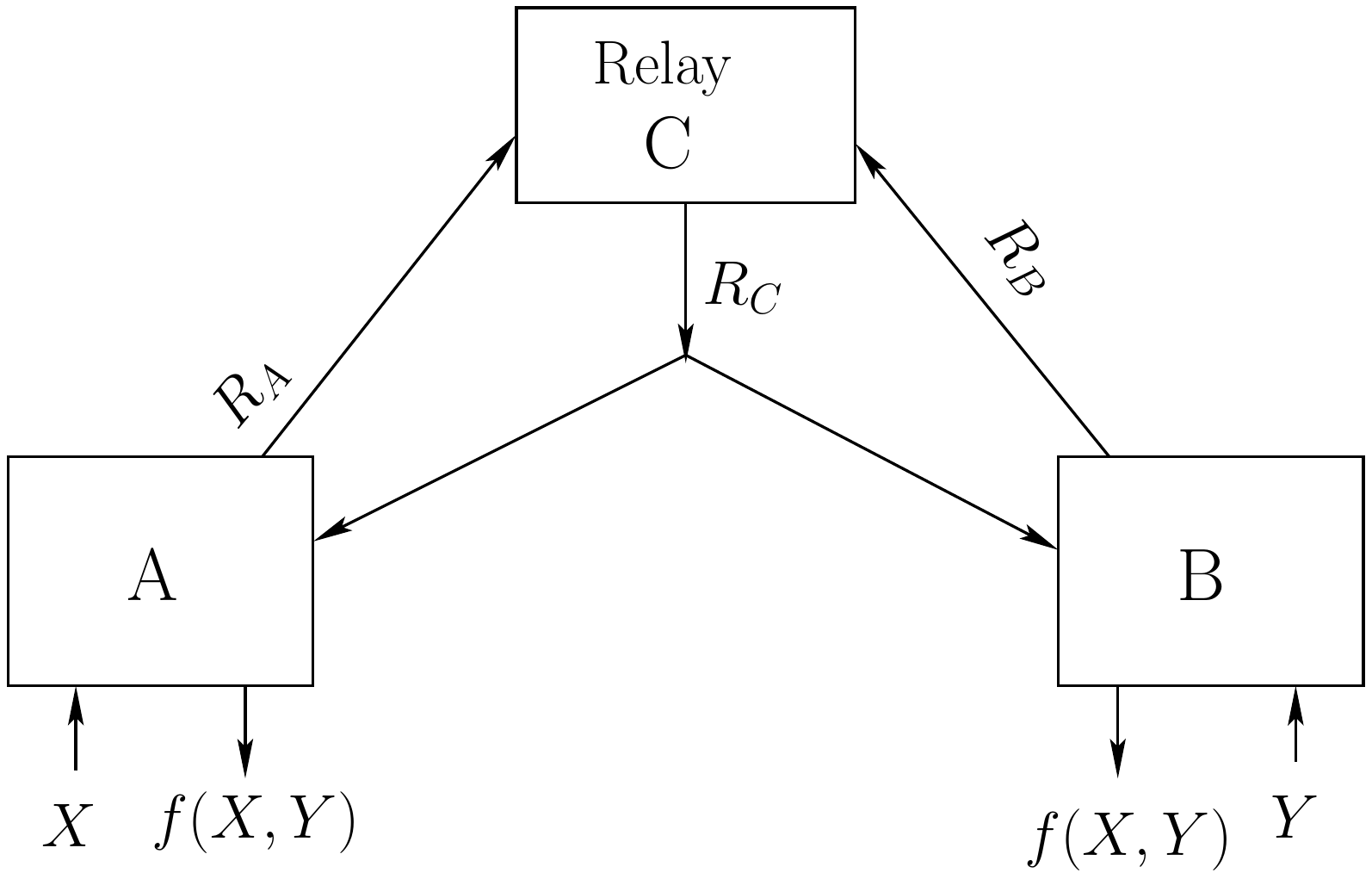}
\caption{Function computation in bidirectional relay network}
\label{Relay_Network}
\end{figure}

The problem of zero error source coding with receiver side information 
was first studied by Witsenhausen in \cite{Witsen_1976}. 
Here, the receiver has side information $Y$ and it wants to recover the
source random variable $X$ with zero error.
A confusability graph and its AND product graphs were defined, and
the minimum rate was characterized in terms of 
the chromatic number of the $n$-times AND product graph.
The same side information problem was later considered in \cite{Alon_1996}
in the context of variable length coding. The goal was to minimize the expected number of bits
per symbol that need to be sent from the transmitter. This quantity was shown to be the limit of the chromatic
entropy of the AND product of the confusability graph. This asymptotic rate was later shown~\cite{Rose_2003} to be the 
complementary graph entropy~\cite{Longo_1973} of the confusability graph.
A single letter characterization for complementary 
graph entropy is still unknown. 

In \cite{Alon_1996}, protocols are further considered for the special case of 
unrestricted inputs, where $(X,Y)$ may take values outside the support
of $p_{XY}$. It was 
shown that the minimal asymptotic rate is the chromatic entropy of the OR 
product graph of the confusability graph. It was further shown that this 
rate can be expressed as the graph entropy~\cite{Korner_1973,Simonyi_1995} 
of the confusability graph - thus resulting in a single-letter characterization. 

Distributed encoding of two correlated sources and joint decoding with
zero error is considered in \cite{Koulgi_2003}. A Single
letter characterization is given for the achievable region under 
unrestricted inputs, and this gives an inner bound for the original problem
under restricted inputs.
Another unidirectinal network is considered in \cite{Shayevitz_2014} for distributed 
computation. 
Here two distribute sources (with $X$ and $Y$) encode 
and send messages to a common relay, which in turn sends a message to 
a decoder with side-information $Z$. The decoder wants to compute
a function $f(X,Y,Z)$ with zero error. The input is assumed to be
unrestricted.  The idea of graph entropy region of a probabilistic 
graph is introduced, and some single letter inner and outer bounds for 
the graph entropy region are given. 

 For our zero-error computation problem depicted in Fig.~\ref{Relay_Network}, we 
provide a characterization of the rate region in terms of graph coloring 
of some suitably defined graphs. We provide single letter inner and 
outer bounds for the rate region. A sufficient condition on the joint 
distribution $p_{XY}$ is identified under which, the relay will
also be able to reconstruct $f(X,Y)$ for any scheme where A and B
reconstruct it with zero error.

The paper is organized as follows. Section~\ref{sec:definitions} presents 
problem formulation and some definitions. In Section~\ref{sec:results}, we 
define graph entropy region and state our main results.
Proof of the results and some examples are given in Section ~\ref{sec:proofs}.
We conclude the paper in Section~\ref{Conclusion}.

\section{Problem formulation and some definitions}
\label{sec:definitions}
\subsection{Problem formulation}

Nodes A and B observe $X$ and $Y$ respectively from finite alphabet
sets $\cX$ and $\cY$. $(X,Y)$ have a joint distribution $p_{XY}(x,y)$,
and their different realizations are i.i.d. In other words, $n$ consecutive
realizations $(X^n,Y^n)$ are distributed as 
$Pr(x^n,y^n)=\prod_{i=1}^{n} p_{XY}(x_i,y_i) $ for all 
$x^n = (x_1,x_2,\cdots,x_n)$ and $y^n = (y_1,y_2,\cdots,y_n)$.

The support set of $(X,Y)$ is defined as
$$S_{XY}= \{ (x,y): p_{XY}(x,y)>0 \}.$$

On observing $X^n$ and $Y^n$ respectively, A and B send messages
$M_A$ and $M_B$ using prefix free codes such that
$E|M_A|=nR_A$ and $E|M_B|=nR_B$. Here $|.|$ denotes the length of the
respective message in bits.
C then broadcasts a message $M_C$ with $E|M_C|=nR_C$
to A and B. Each of A and B then decode $f(X_i,Y_i);\;i=1,2,\cdots,n$ from
the information available to them. A length-$n$ scheme
is a quintuple $(\phi_{A},\phi_{B},\phi_{C},\psi_A,\psi_B)$,
where

$$ \phi_{A}: \cX^{n} \longrightarrow \{ 0,1\}^{*}, \quad
\phi_{B}: \cY^{n} \longrightarrow \{ 0,1\}^{*}$$
\mbox{and } 
$$\phi_{C}:\phi_{A}(\cX^{n}) \times \phi_{B}(\cY^{n}) \longrightarrow \{ 0,1\}^{*}$$

are prefix free encoding functions of A, B and C respectively, and

$$\psi_A: \cX^n\times \phi_{C}\left(\phi_{A}(\cX^{n}) \times \phi_{B}(\cY^{n})\right) \longrightarrow \cZ^n$$
and
$$\psi_B: \cY^n \times \phi_{C}\left(\phi_{A}(\cX^{n}) \times \phi_{B}(\cY^{n})\right)  \longrightarrow \cZ^n$$
are the decoding functions of A and B. Here $\{ 0,1\}^{*}$ denotes the set of 
all finite length binary sequences. Let $(\psi_A(\cdot))_i$ and 
$(\psi_B(\cdot))_i$ denote
the $i$-th components of $\psi_A(\cdot)$ and $\psi_B(\cdot)$ respectively.
A scheme is called a {\it zero-error scheme}
if for each $X^n\in \cX^n, Y^n\in \cY^n$, and $i=1,2,\cdots,n$,
\begin{align}
& (\psi_A(X^n,\phi_C(\phi_A(X^n),\phi_B(Y^n))))_i = f(X_i,Y_i)\notag\\
\mbox{and }\notag\\
& (\psi_B(Y^n,\phi_C(\phi_A(X^n),\phi_B(Y^n))))_i = f(X_i,Y_i)\notag
\end{align}
if $(X_i,Y_i) \in S_{XY}$.

The rate triplet $(R_{A}, R_{B}, R_{C})$ of a scheme is defined as
\begin{eqnarray*}
R_{A} & = & \frac{1}{n} \sum_{x^{n}} Pr(x^{n}) \mid \phi_{A}(x^{n}) \mid \\
R_{B} & = & \frac{1}{n} \sum_{y^{n}} Pr(x^{n}) \mid \phi_{B}(y^{n}) \mid \\
R_{C} & = & \frac{1}{n} \sum_{(x^{n},y^{n})} Pr(x^{n},y^{n}) \mid \phi_{C}(\phi_{A}(x^{n}), \phi_{B}(y^{n})) \mid.
\end{eqnarray*}
A rate-triple is said to be achievable if there is a zero-error scheme
of some length with that rate-triple.
The rate-region $\cR(f,X,Y)$ is the closure of the set of achievable
rate-triples.

\subsection{Graph theoretic definitions}

Let $G$ be a graph with vertex set $V(G)$ and edge set $E(G)$. A set 
$S \subseteq V(G) $ is called an
independent set if no two vertices in $S$ are adjacent in $G$.
The $n$-fold OR product of $G$, denoted by $G^{\vee n}$,
is defined by $V(G^{\vee n}) = (V(G))^n$ and $E(G^{\vee n})=
\{(v^n,v'^n): (v_i,v'_i)\in E(G) \mbox{ for some } i\}$.

For a graph $G$ and a random variable $X$ taking values in $V(G)$,
$(G,X)$ represents a {\it probabilistic graph}.
Chromatic entropy~\cite{Alon_1996} of $(G,X)$ is defined as
\begin{align*}
 H_{\chi}(G,X) &= \mbox{min} \{H[c(X)]: \: c  \mbox{ is a coloring of } G \}.
\end{align*}
Let $W$ be distributed over the power set $2^{\cX}$. 
The graph entropy of the probabilistic graph $(G,X)$ is defined as
\begin{align}
 H_G(X) = \min_{X\in W \in \Gamma(G)} I(W;X),
 \label{eq:gentropy}
\end{align}
where $\Gamma(G)$ is the set of all independent sets of $G$. Here the minimum is taken over all 
conditional distribution $p_{W|X}$ which is non-zero only for $X\in W$.
The following interesting result was shown in
\cite{Alon_1996}.
\begin{align}
\lim\limits_{n\to\infty} \frac{1}{n} H_{\chi}(G^{\vee n}, X^n) = H_G(X).
\label{eq:gpentropy}
\end{align}




\comment{
\begin{table} 
\centering
\begin{tabular}{|l|l|}
\hline
$x$ & $\phi_A(x)$  \\\hline
$0$ & R  \\
$1$ & B  \\
$2$ & R  \\
$3$ & B  \\
$4$ & G  \\\hline
\end{tabular}
\quad
\begin{tabular}{|l|l|}
\hline
$y$ & $\phi_Y(y)$  \\\hline
$0$ & R  \\
$1$ & B  \\
$2$ & R  \\
$3$ & B  \\
$4$ & G  \\\hline
\end{tabular}
\quad
\begin{tabular}{|l|l l l|}
\hline
\diaghead{\theadfont Fontsize increas}%
{$\phi_A(x)$}{$\phi_B(y)$}&
\thead{R}&\thead{B}&\thead{G}\\ \hline
R & Y & C &  \\
B & C & Y & C \\   
G & C &  & Y  \\\hline
\end{tabular}\\
\caption{Mappings by Node A,B and C}
\label{Mappings}
\end{table}
}

We now define some graphs suitable for addressing our problem. 
For a function $f(x,y)$ defined over $\cX \times \cY$, we define a graph called 
$f$-modified rook's graph. 
A rook's graph $G_{XY}$ over $\cX \times \cY$ is defined by the  
vertex set  $\cX \times \cY$ and edge set
$\{((x,y),(x',y')):x=x' \mbox{ or } y=y', \mbox{ but } (x,y)\neq (x',y')\}$.

\begin{definition}
For a function $f(x,y)$ the $f$-modified rook's graph 
$G_{XY}^{f}$ has its vertex set $S_{XY}$, and
two vertices $(x_{1}, y_{1})$ and $(x_{2}, y_{2})$ 
are adjacent if and only if 
they are adjacent in the rook's graph 
$G_{XY}$ and $f(x_{1},y_{1}) \neq f(x_{2},y_{2})$. 
\end{definition}

\begin{figure}[h]
 \centering
 \begin{subfigure}[b]{0.25\textwidth}
\includegraphics[scale =0.3]{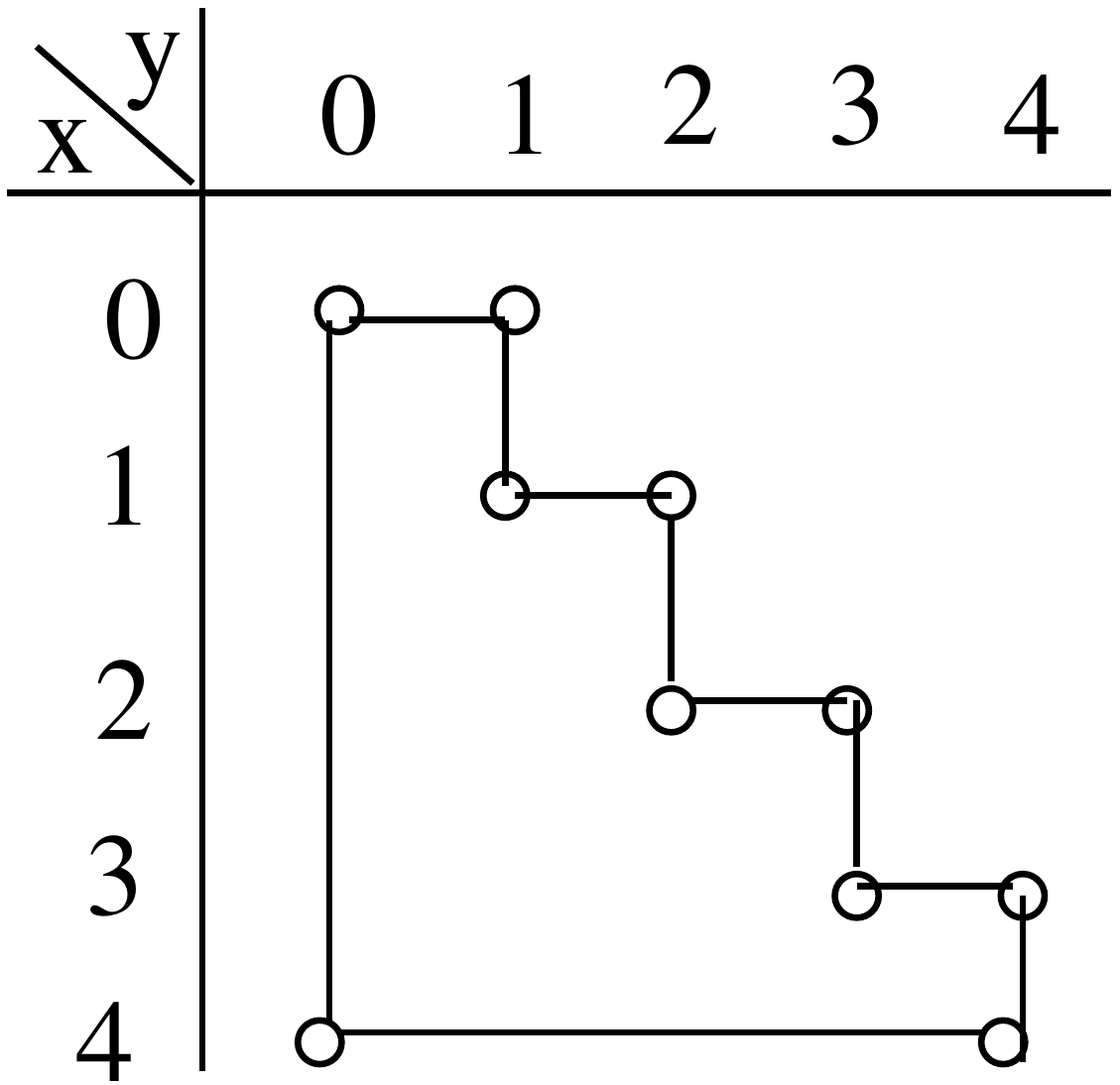}
\caption{$f$-modified rook's graph for $f(x,y)$ in \eqref{eq:function}}
\label{Rook_graph}
\end{subfigure}
\quad
 \begin{subfigure}[b]{0.2\textwidth}
\includegraphics[scale=0.3]{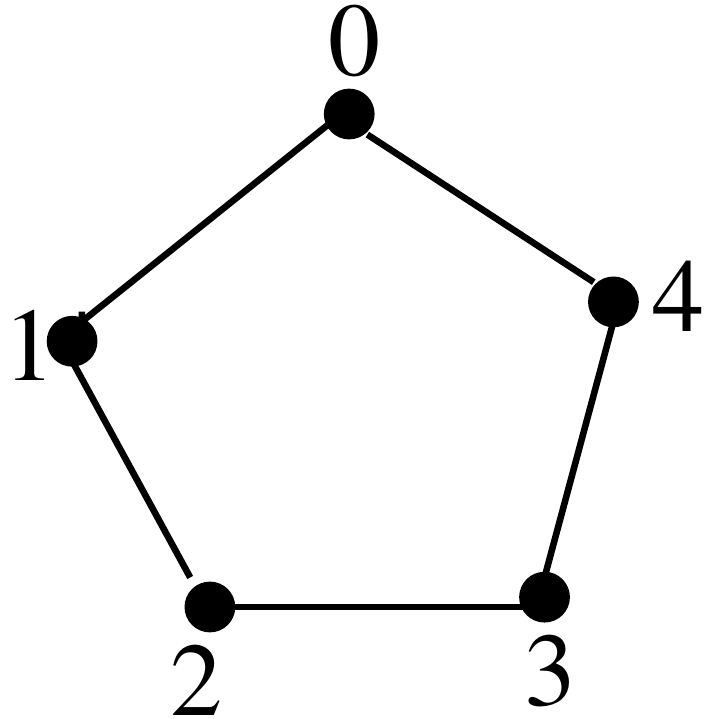}
\caption{$f$-confusability graphs $G_{X|Y}^f,G_{Y|X}^f$ for $f(x,y)$ in \eqref{eq:function}}
\label{Pentagon_Graph}
\end{subfigure}
\caption{$f$-modified rook's graph and $f$-confusability graph}
\label{Fig2}
\end{figure}

For example, let us consider $X,Y \in \{0,1,2,3,4\}$ with distribution
\begin{equation*}
p(x,y) = \left\{
\begin{array}{cl}
  \frac{1}{10} &  \quad \mbox{if} \; y=x \mbox{ or } y=x+1 \mbox{ mod }5\\
              0  & \quad \mbox{otherwise }
\end{array} \right. ,
\end{equation*}
and the equality function
\begin{equation}
f(x,y) = \left\{
\begin{array}{cl}
 1 & \quad \mbox{if} \; x=y\\
            0 & \quad \mbox{if} \; x\neq y.
\end{array} \right. 
\label{eq:function}
\end{equation} 
The $f$-modified rook's graph for this function is shown in Fig.~\ref{Rook_graph}.

$f$-confusability graph $G_{X|Y}^f$ of $X,Y$ and $f$ was defined in \cite{Shayevitz_2014}.
Its vertex set is $\cX$, and two vertices $x$ and $x'$ are adjacent if and only if
$\exists \; y\in \cY$ such that
$f(x,y) \neq f(x',y)$ and $(x,y),(x',y) \in S_{XY}$. $G_{Y|X}^f$
is defined similarly. For the function defined in~\eqref{eq:function}, $G_{X|Y}^f$ and $G_{Y|X}^f$
are the same graph which is shown in Fig.~\ref{Pentagon_Graph}.

\comment{
\begin{table}
 \centering

\begin{tabular}{|c|c c c c c|}
\hline
\diaghead{\theadfont Fontsize 1}%
{x}{y}
& 0 & 1 & 2 & 3 & 4 \\\hline
0& (R,R) & (R,B) &  & & \\
1&  & (B,B) & (B,R) &  & \\
2&  &  & (R,R) & (R,B) & \\
3&  &  &  & (B,B) & (B,G) \\
4& (G,R) &  &  & & (G,G) \\\hline
\end{tabular}
\caption{Coloring that relay receives}
\vspace{2 mm}

\begin{tabular}{|c|c c c c c|}
\hline
\diaghead{\theadfont Fontsize 1}%
{x}{y}
& 0 & 1 & 2 & 3 & 4 \\\hline
0& Y & C &  & & \\
1&  & Y & C &  & \\
2&  &  & Y & C & \\
3&  &  &  & Y & C \\
4& C &  &  & & Y \\\hline
 \end{tabular}
\caption{Coloring that relay transmits}
\label{Coloring}
\end{table}
}


\section{Summary of Results}
\label{sec:results}
\subsection{Characterization of $\cR(f,X,Y)$}

\comment{
In this section we give a characterization of rate region $\cR(f,X,Y)$ and state our main results.
Characterization of $\cR(f,X,Y)$ is based on the coloring of $f$-modified rook's graph $G_{XY}^f$.
We first provide an example of encoding functions used by nodes A, B and C for the function shown in Table~\ref{Pentagon_Function}.
Table~\ref{Mappings} shows the three mappings used by nodes A, B and C for the function when $n=1$.
Table~\ref{Coloring} shows the pair of colors that relay receives and the 
the color that relay transmits for these mappings..
The encoder pair $(\phi_A,\phi_B)$ constitutes a coloring of the $f$-modified 
rook's graph for the function. The mapping used by relay $\phi_C$ is also a 
coloring of that particular graph. 
Here observe that the pairs $\left(\phi_C(\phi_A(x),\phi_B(y)),x\right)$ and $\left(\phi_C(\phi_A(x),\phi_B(y)),y\right)$
can compute $f(x,y)$ for any $(x,y) \in S_{XY}$.
This shows that $(\log 3,\log 3, \log 2)$ is in the
achievable rate region for this example.
Now we define graph entropy region for characterizing the rate region $\cR(f,X,Y).$
}

We first define the chromatic entropy region of a $f$-modified rooks graph.
If $c_1$ and $c_2$ are two maps of $\cX$ and $\cY$ respectively, then
$c_1\times c_2$ denotes the map given by $(c_1\times c_2)(x,y) = 
(c_1(x),c_2(y))$.

Recall that $S_{XY}$ is the vertex set of $G_{XY}^f$.
A triplet $(c_A,c_B,c_C)$ of functions defined 
over $\cX, \cY, S_{XY}$ respectively is called a {\it color cover} for 
$G_{XY}^f$ if
\begin{enumerate}[i)]
 \item $c_A \times c_B$ and $c_C$ are colorings of $G_{XY}^f$.
 \item $c_A \times c_B$ is a refinement of $c_C$ in $S_{XY}$, i.e., $\exists$ a
mapping $\theta:(c_A\times c_B)(S_{XY}) \rightarrow c_C(S_{XY})$ such
that $\theta \circ (c_A\times c_B) = c_C$.
\end{enumerate}
{\it Chromatic entropy region} $R_{\chi}( G_{XY}^f,X,Y)$ of  $G_{XY}^f$ is defined as
\begin{align*}
  & R_{\chi}( G_{XY}^f,X,Y) \triangleq \bigcup_{(c_A,c_B,c_C)}\{(b_A,b_B,b_C): \\
  & \hspace{5 mm} b_A \geq H(c_A(X)), b_B \geq H(c_B(Y)), b_C \geq H(c_C(X,Y))  \}, 
\end{align*}
where the union is taken over all color covers of $G_{XY}^f$. 
Motivated by the result \eqref{eq:gpentropy}, we define the
{\it graph entropy region} as
\begin{equation*}
HR_{G_{XY}^f}(X,Y) \triangleq \bigcup_n \frac{1}{n} R_{\chi}\left( (G_{XY}^f)^{\vee n}, X^n,Y^n\right).
\end{equation*}	
%
\begin{theorem}
 \begin{enumerate}[(i)]
  \item $\cR(f,X,Y) = HR_{G_{XY}^f}(X,Y).$ \label{Part1}
  \item Let \label{Part2}
\begin{align*}
 \cR_{I1} \triangleq & \{(R_A,R_B,R_C) : R_A \geq H(X), R_B \geq H(Y), \nonumber\\
 & \hspace{25mm} R_C \geq H_{G_{XY}^f}(X,Y) \}\nonumber\\ 
 \cR_{I2} \triangleq & \{(R_A,R_B,R_C) :R_A \geq H_{G_{X|Y}^f}(X), \nonumber\\
 &  R_B \geq H_{G_{Y|X}^f}(Y),\\
 &  R_C \geq H_{G_{X|Y}^f}(X)+ H_{G_{Y|X}^f}(Y)\}.
  \end{align*} 
 Let $\cR_I$ be the convex hull of $\cR_{I1} \cup \cR_{I2}$. Then\\
   $\cR_I \subseteq \cR(f,X,Y)  .$
\item Let 
 \begin{align*}
  \cR_O \triangleq  & \{(R_A,R_B,R_C) :R_A \geq H_{G_{X|Y}^f}(X), \nonumber \\
 & R_B \geq H_{G_{Y|X}^f}(Y), R_C \geq H_{G_{XY}^f}(X,Y) \}.
 \end{align*}
Then $\cR(f,X,Y) \subseteq \cR_O.$\label{Part3}
\item If $G_{Y|X}^f$ and $G_{Y|X}^f$ are complete graphs, then $\cR_{I}=\cR_{I1} = \cR_O$. \label{Part4}
 \end{enumerate}
 \label{Rate_Region}
\end{theorem}
%

\begin{theorem}
 If $p(x,y)>0 \quad \forall \, (x,y) \in \cX \times \cY$, then for any 
zero-error scheme the relay can compute $f(x,y)$ for all $(x,y) \in 
\cX\times \cY$.
\label{Relay_function}
\end{theorem}

\section{Proofs of the results}
\label{sec:proofs}
To prove Theorem~\ref{Rate_Region}, we first present some lemmas.
\begin{lemma}
 For $n=1$, and given the encoding functions $\phi_A,\phi_B,\phi_C$, 
the nodes A and B can recover $f(X,Y)$ with zero error if and only if
$\phi_C \circ (\phi_A \times \phi_B)$ is a coloring of 
$G_{XY}^f$.
 \label{Relay_mapping_n_1}
\end{lemma}

\begin{proof} 
Let $E(G_{XY}^f)$ denote the set of edges of $G_{XY}^f$.
Note that 
\begin{align}
 E(G_{XY}^f) & = \{((x,y),(x,y')) \in S_{XY}| f(x,y) \neq f(x,y')\}\notag \\
& \cup \{((x,y),(x',y)) \in S_{XY}| f(x,y) \neq f(x',y)\}
\label{eq:edgeset}
\end{align}
Note that each edge is of the form $((x,y),(x,y'))$ or $((x,y),(x',y))$.

A and B can recover $f(X,Y)$ with zero error $\Leftrightarrow$ (i)
for any $(x,y),(x,y')\in S_{XY}$ with
$f(x,y) \neq f(x,y')$, $\phi_C(\phi_A(x),\phi_B(y)) \neq 
\phi_C(\phi_A(x),\phi_B(y'))$ and (ii) for any $(x,y),(x',y)\in S_{XY}$ with
$f(x,y) \neq f(x',y)$, $\phi_C(\phi_A(x),\phi_B(y)) \neq 
\phi_C(\phi_A(x'),\phi_B(y))$ 
$\Leftrightarrow$ 
for any $((x,y),(x',y')) \in E(G_{XY}^f)$, $\phi_C(\phi_A(x),\phi_B(y)) \neq
\phi_C(\phi_A(x'),\phi_B(y'))$ $\Leftrightarrow$ $\phi_C \circ (\phi_A \times \phi_B)$ is a coloring of 
$G_{XY}^f$.
\end{proof}

Clearly, whenever $\phi_C\circ (\phi_A\times \phi_B)$ is a coloring of 
$G_{XY}^f$,  $\phi_A\times\phi_B$ is also a coloring of 
$G_{XY}^f$. The following lemma gives a necessary and sufficient
condition for $(\phi_A,\phi_B)$ to be a coloring of $G_{XY}^f$.

\begin{lemma}
 $(\phi_A,\phi_B)$ is a coloring of $G_{XY}^f$ if and only if $\phi_A$ is a coloring 
 of $G_{X|Y}^f$ and $\phi_B$ is a coloring of  $G_{Y|X}^f$.
 \label{Nodes_mapping_n_1}
\end{lemma}
\begin{proof}
$\phi_A$ and $\phi_B$ are colorings of $G_{X|Y}^f$ and $G_{Y|X}^f$ respectively
$\Leftrightarrow$ (i)
 any $(x,y), (x,y') \in S_{XY}$ with
$f(x,y) \neq f(x,y')$, satisfies $(\phi_A(x),\phi_B(y))\neq (\phi_A(x),\phi_B(y'))$ and (ii) any $(x,y), (x',y) \in S_{XY}$
with $f(x,y) \neq f(x',y)$ satisfies $(\phi_A(x),\phi_B(y))\neq 
(\phi_A(x'),\phi_B(y))$ $\Leftrightarrow$ $\phi_A \times \phi_B$ is a 
coloring of $G_{XY}^f$. We used \eqref{eq:edgeset} in the last equivalence.
\end{proof}

We now extend the above lemmas for $n$-length schemes.

\begin{lemma}
 For any $n$, and given the encoding functions $\phi_A,\phi_B,\phi_C$, 
the nodes A and B can recover $f(X_i,Y_i);\; i=1,2,\cdots,n$ with zero error if and only if
$\phi_C \circ (\phi_A \times \phi_B)$ is a coloring of 
$(G_{XY}^f)^{\vee n}$.
 \label{Relay_mapping_n}
\end{lemma}
\begin{proof}
 Consider some $x^n \in \cX^n$ and $y^n, y'^n \in \cY^n$ 
 such that for some $i$, $p(x_i,y_i).p(x_i,y_i') >0$ and
 $f(x_i,y_i) \neq f(x_i,y_i')$. For pairs $(x^n,y^n)$ and $(x^n,y'^n)$ 
 node A should receive different data from the relay 
 for zero error computation in the unrestricted setup. 
 Similarly node B should receive different data for pairs $(x^n,y^n)$ 
and $(x'^n,y^n)$ if for some $i$,
 $f(x_i,y_i) \neq f(x_i',y_i)$ and $(x_i,y_i),(x_i',y_i)\in S_{XY}$. 
 This shows that zero-error computation is possible if and only if 
$\phi_C\circ (\phi_A\times\phi_B)$ is a coloring of $(G_{XY}^f)^{\vee n}$.
\end{proof}

\begin{lemma}
 $\phi_A\times\phi_B$ is a coloring of $(G_{XY}^f)^{\vee n}$ if and only if $\phi_A$ is a coloring 
 of $(G_{X|Y}^f)^{\vee n}$ and $\phi_B$ is a coloring of $(G_{Y|X}^f)^{\vee n}$.
 \label{Nodes_mapping_n}
\end{lemma}
\begin{proof}
If $\phi_A$ is not a coloring of $(G_{X|Y}^f)^{\vee n}$, then $\exists \:
x^n, x'^n \in \cX^n$ and $y^n \in \cY^n$ with 
$(x_i,y_i),(x_i',y_i) \in S_{XY}, f(x_i,y_i) \neq f(x_i',y_i)$ for some
$i$ such that $\phi_A(x^n) = \phi_A(x'^n)$. Then $((x^n,y^n),(x'^n,y^n))
\in E(G_{XY}^f)^{\vee n}$, and $(\phi_A(x^n),\phi_B(y^n))
=(\phi_A(x'^n),\phi_B(y^n))$. So $\phi_A\times \phi_B$
is not a coloring of 
$(G_{XY}^f)^{\vee n}$. Similarly, if $\phi_B$ is not a coloring
of $(G_{Y|X}^f)^{\vee n}$, then $\phi_A\times\phi_B$ is not a coloring 
of $(G_{XY}^f)^{\vee n}$.
The opposite implication also follows along similar lines using \eqref{eq:edgeset}.
 \end{proof}

\begin{lemma}
  \label{Lem_Outer_n}
For a given $n$, the rate-triple of any length-$n$ zero-error scheme satisfies
 \begin{eqnarray*}
  R_A \geq  \frac{1}{n}H_{\chi}\left((G_{X|Y}^f)^{\vee n},X^n\right)\\
  R_B \geq  \frac{1}{n}H_{\chi}\left((G_{Y|X}^f)^{\vee n},Y^n\right)\\
  R_C \geq  \frac{1}{n} H_{\chi}\left((G_{XY}^f)^{\vee n},(X^n,Y^n)\right).
 \end{eqnarray*}
\end{lemma}
\begin{proof}
 This follows from the definition of chromatic entropy, Lemmas \ref{Relay_mapping_n} and \ref{Nodes_mapping_n}.
\end{proof}
\subsection{Proof of Theorem~\ref{Rate_Region}}
\noindent
\begin{proof}
{\it Proof of Part~(\ref{Part1})}: Lemma \ref{Relay_mapping_n}
implies that for encoding functions $\phi_A,\phi_B,\phi_C$ of 
any zero-error scheme, $\phi_A, \phi_B, \phi_C\circ (\phi_A\times\phi_B)$
is a color cover for $(G_{XY}^f)^{\vee n}$.
Similarly, for any  color cover $(c_A,c_B,c_C)$ of 
$(G_{XY}^f)^{\vee n}$, let $\phi_A,\phi_B$ be any prefix-free 
encoding functions of $c_A$ and $c_B$ respectively. 
Since $c_A\times c_B$ is a refinement of $c_C$,  there exists a
mapping $\theta_C$ such that $c_C = \theta_C \circ (c_A\times c_B)$.
Taking $\phi_C$ as any prefix-free encoding of $c_C$ yields
a scheme with encoding functions $(\phi_A,\phi_B,\phi_C)$.
Thus the result follows from the definition of graph entropy region 
$HR_{G_{XY}^f}(X,Y)$. \\

\noindent
{\it Proof of Part~(\ref{Part2})}:
Let us consider a {\it zero-error scheme} 
in which nodes A and B 
communicate $X^n$ and $Y^n$ to the relay node C. 
On knowing $X^n$ and $Y^n$, the relay broadcasts a minimum entropy coloring of
$(G_{XY}^f)^{\vee n}$ to A and B.
Using an optimum prefix free code for each communication, 
the achieved rates satisfy
 \begin{align*}
  R_A & \leq \frac{1}{n} H(X^n)+\frac{1}{n}\\
  R_B & \leq \frac{1}{n} H(Y^n)+\frac{1}{n} \\
\mbox{and }\\
R_C & \leq \frac{1}{n} H_{\chi}\left((G_{XY}^f)^{\vee n},(X^n,Y^n)\right) +\frac{1}{n}.
 \end{align*}
Taking $n\rightarrow \infty$, this proves that any point in $\cR_{I1}$ is achievable.

Now let us consider another scheme. Node A and B send a 
coloring of $(G_{X|Y}^f)^{\vee n}$ and $(G_{Y|X}^f)^{\vee n}$
respectively to C. C then broadcasts both these colors to A and B.
Again, using optimum prefix free zero error codes, the rates achieved satisfy
\begin{align*}
& R_A  \leq \frac{1}{n} H_{\chi}\left((G_{X|Y}^f)^{\vee n},X^n\right) +\frac{1}{n} \\
& R_B  \leq \frac{1}{n} H_{\chi}\left((G_{Y|X}^f)^{\vee n},Y^n\right) +\frac{1}{n}\\
\mbox{and }\\
R_C  \leq & \frac{1}{n}\left(H_{\chi}((G_{X|Y}^f)^{\vee n},X^n)
                  + H_{\chi}((G_{Y|X}^f)^{\vee n},Y^n) +2\right).
\end{align*}
Thus any point in $\cR_{I2}$ is achievable. Thus any point in
the convex hull $\cR_I$ of $\cR_{I1} \cup \cR_{I2}$ is also achievable
by time-sharing.\\

\noindent
{\it Proof of Part~(\ref{Part3})}:
As $n \rightarrow \infty$, Lemma~\ref{Lem_Outer_n} and \eqref{eq:gpentropy}
give the outer bound $\cR_O$.\\

\noindent
{\it Proof of Part~(\ref{Part4})}:
When $G$ is a complete graph, the only independent 
sets are the singletons of $V(G)$. So for $(G,X)$ with such a graph, 
the only choice of $W$ in \eqref{eq:gentropy} is $W=\{X\}$ and
\begin{equation*}
 H_G(X) = I(W;X) = H(X).
\end{equation*}
Thus if $G_{Y|X}^f$ and $G_{Y|X}^f$ are complete graphs, then $H_{G_{X|Y}^f}(X) = H(X)$
and $H_{G_{Y|X}^f}(Y) = H(Y)$. Then $\cR_{I1} = \cR_O$, and thus the result
follows.
\end{proof}
In the following, we provide an example of a function for which $\cR_I = \cR_O$.
\begin{example}
Let $X$ and $Y$ be uniformly distributed over $\{0,1,2\}$ and let $S_{XY} = \cX \times \cY$.
Nodes A and B want to compute $min(X,Y)$. For this example, the $f$-confusability graphs
$G_{X|Y}^f$ and $G_{Y|X}^f$ are complete graphs with vertex set $\{ 0,1,2\}$. 
So here we get $\cR_{I}=\cR_{I1} = \cR_O$.
\end{example}


 \subsection{Proof of Theorem~\ref{Relay_function}}
\begin{proof}
Let us consider any $n$. With abuse of notation, we denote the 
messages sent by the nodes A, B, and C by $\phi_A,\phi_B,\phi_C$ 
respectively. 
In the following, we omit the arguments, and denote $f(X^n,Y^n)$ by simply $f$.
Since the function is computed with zero error at nodes A and B, we have
$H(f|\phi_C,X^n)=0$ and $H(f|\phi_C,Y^n)=0$. 
We want to show that $H(f|\phi_A, \phi_B)=0$.
We prove this by contradiction. Let us assume that $H(f|\phi_A,\phi_B)>0$. 
Then $\exists$ $(x^n,y^n)$ and $(x'^n,y'^n)$ such that 
\begin{align}
& Pr(X=x^n,Y=y^n, \phi_A=k_1,\phi_B=k_2)>0 \label{prob_1}\\
& Pr(X=x'^n,Y=y'^n, \phi_A=k_1,\phi_B=k_2)>0 \label{prob_2}\\
\mbox{and } & f(x_i,y_i)\neq f(x'_i,y'_i) \mbox{ for some } i. \label{eq:fn1}
\end{align}
We consider two cases.

Case 1: $x_i=x'_i=x$.  
Since we have $Pr(X^n=x^n, \phi_A=k_1)>0$ (using \eqref{prob_1}), we get 
$Pr(X^n=x^n,Y^n=y'^n, \phi_A=k_1)= Pr(X^n=x^n, \phi_A=k_1)Pr(Y^n=y'^n|X^n=x^n)>0$. So we get 
\begin{equation}
 Pr(X^n=x^n,Y=y'^n, \phi_A=k_1,\phi_B=k_2)>0
 \label{prob_3}
\end{equation}
as $\phi_B(y'^n) = k_2$.

Taking $k_0=\phi_C(k_1,k_2)$, \eqref{prob_1}, \eqref{prob_3}
imply that
$Pr(X^n=x^n,Y^n=y^n, \phi_C=k_0)>0$ and $Pr(X^n=x^n,Y^n=y'^n, \phi_C=k_0)>0$.

This, together
with \eqref{eq:fn1} gives $H(f|\phi_C,X^n)>0$. Thus A can not recover $f(X,Y)$
with zero error - a contradiction.\\

Case 2: $x_i \neq x'_i$ and $y_i\neq y'_i$. Using \eqref{eq:fn1}, we get
either $f(x_i,y'_i) \neq f(x_i,y_i)$ or $f(x_i,y'_i) \neq f(x'_i,y'_i)$.
W.l.o.g, let us assume $f(x_i,y'_i) \neq f(x_i,y_i)$.
Then by combining \eqref{prob_3} and \eqref{prob_1}, and using the fact that
$f(x_i,y'_i) \neq f(x_i,y_i)$, we get $H(f|X^n,\phi_C) \neq 0$.
Thus A can not recover $f(X^n,Y^n)$ with zero error - a contradiction.
This completes the proof of the theorem.
\end{proof}

Theorem~\ref{Relay_function} does not hold if $S_{XY} \neq \cX\times \cY$. In
the following, we consider a simple example to demonstrate this.
Here nodes A and B recover the function with zero error, but the
relay can not reconstruct the function.
\begin{example}
Consider $X,Y \in \{1,2,3\}$
\begin{equation*}
 p(x,y) =\left\{
\begin{array}{cl}
  \frac{1}{6} & \quad \mbox{if} \; x\neq y\\
          0 & \quad \mbox{otherwise}
\end{array} \right.
\end{equation*}
and
\begin{equation*}
f(x,y) = \left\{
\begin{array}{cl}
 1 & \quad \mbox{if} \; x>y\\
            0 & \quad \mbox{if} \; x \leq y.
\end{array} \right. 
\end{equation*} 
Let $\phi_A,\phi_B$ and $\phi_C$ be as follows.
\begin{equation*}
  \phi_A =\left\{
\begin{array}{cl} 
  1 & \quad \mbox{if} \; x=1\\
         0 & \quad \mbox{otherwise.}
         \end{array} \right. 
\end{equation*}

\begin{equation*}
  \phi_B  =\left\{
\begin{array}{cl} 
  1 & \quad \mbox{if} \; y=1\\
         0 & \quad \mbox{otherwise.}
\end{array} \right. 
\end{equation*}

\begin{equation*}
 \phi_C =\left\{
\begin{array}{cl} 
 
  1 & \quad \mbox{if} \quad  \phi_A=\phi_B\\
         0 & \quad \mbox{otherwise.}
 \end{array} \right.        
\end{equation*}

In this example the relay cannot determine the function value in all the cases. When $\phi_A = \phi_B =0$, the function
value can be both $0$ and $1$. So here $H(f|\phi_A,\phi_B)>0$.
\end{example}
\section{Conclusion}
\label{Conclusion}
We provided a characterization of the rate region for our function computation problem
in terms of graph coloring and established single letter inner and outer bounds of the rate region.
A sufficient condition on the $f$-confusability graphs is identified 
under which these inner and outer bounds coincide.
We also showed
that if $p_{XY}$ is non-zero for all pairs of values, then the relay
can compute the function if both A and B can compute it.
We addressed the problem only for one-round protocols.  
Investigating the problem under multi-round protocols is an interesting
direction of future work.

\section*{Acknowledgment}
The work  was supported in part by the Bharti Centre for Communication, IIT Bombay and 
a grant from the Information Technology Research Academy, Media Lab Asia, to IIT Bombay.

\end{document}